\newif\ifdouble

\doubletrue

\ifdouble
  \documentclass[journal]{IEEEtran}
\else
  \documentclass[journal,draftclsnofoot,onecolumn]{IEEEtran}
\fi

\usepackage[utf8]{inputenc}
\usepackage{graphicx} 
\usepackage{amsmath} 
\usepackage{amssymb}
\usepackage{amsthm}
\usepackage{color}

\graphicspath{{Images/}}
\newcommand{\R}{\mathbb{R}}
\newcommand{\N}{\mathbb{N}}

\newtheorem{thm}{Theorem}
\newtheorem{defn}[thm]{Definition}
\newtheorem{prop}[thm]{Proposition}
\newtheorem{remark}{Remark}
\newtheorem{lemma}[thm]{Lemma}
\newtheorem{corollary}[thm]{Corollary}

\newtheorem{assum}{Assumption}

\title{Compositional abstraction and safety synthesis\\ using overlapping symbolic models}
\author{Pierre-Jean~Meyer,
        Antoine~Girard,
        and~Emmanuel~Witrant 
\thanks{P.-J. Meyer is with KTH Royal Institute of Technology, Department of Automatic Control, 10044 Stockholm, Sweden (email: pjmeyer@kth.se).}
\thanks{A. Girard is with Laboratoire des signaux et syst\`emes (L2S), CNRS, CentraleSup\'elec, Universit\'e Paris-Sud, Universit\'e Paris-Saclay, 3, rue Joliot-Curie, 91192 Gif-sur-Yvette, cedex, France (email: Antoine.Girard@l2s.centralesupelec.fr).
His work was partially supported by the laboratory of excellence DigiCosme
(CODECSYS project).
}
\thanks{E. Witrant is with Univ. Grenoble Alpes/CNRS, GIPSA-Lab, F-38000 Grenoble, France (email: Emmanuel.Witrant@ujf-grenoble.fr).}
}

\begin{document}

\maketitle

\begin{abstract}
In this paper, we develop a compositional approach to abstraction and safety synthesis for a general class of discrete time nonlinear systems.
Our approach makes it possible to define a symbolic abstraction by composing a set of symbolic subsystems that are overlapping in the sense that they can share some common state variables.
We develop compositional safety synthesis techniques using such overlapping symbolic subsystems.
Comparisons, in terms of conservativeness and of computational complexity, between abstractions and controllers obtained from different system decompositions are provided.
Numerical experiments show that the proposed approach for symbolic control synthesis enables a significant complexity reduction with respect to the centralized approach, while reducing the conservatism with respect to compositional approaches using non-overlapping subsystems.
\end{abstract}


\section{Introduction}
\label{sec intro}

Symbolic control deals with the use of discrete synthesis techniques for controlling complex continuous or hybrid systems~\cite{belta2007symbolic,tabuada2009symbolic}.
In such approaches, one relies on symbolic abstractions of the orignal system; i.e. dynamical systems with finitely many state and input values, each of which symbolizes  
sets of states and inputs of the concrete system~\cite{alur2000discrete}. This enables the use of discrete controller synthesis techniques, such as supervisory control~\cite{cassandras2009introduction}  or algorithmic game theory~\cite{bloem2012synthesis}, which allows us to address high-level specifications such as safety, reachability or more general properties specified by automata or temporal logic formula~\cite{baier2008principles}. When the behaviors of the concrete system and of its abstraction are related by some formal inclusion relationship (such as alternating simulation~\cite{tabuada2009symbolic} or feedback refinement relations~\cite{reissig2016}), the discrete controller of the abstraction can be refined to control the concrete system, with guarantees of correctness. 

Several approaches exist for computing symbolic abstractions for a wide range of dynamical systems (see e.g.~\cite{tabuada2006linear,pola2008approximately,zamani2012symbolic,zamani2014symbolic,coogan2015mixed,reissig2016}), based on partitions or discretizations of the state and input spaces. The numbers of symbolic states and inputs are then typically exponential in the  dimension of the concrete state and input spaces, respectively.
This limits the application of these approaches to low-dimensional systems.
Several works have been done for improving the scalability of symbolic control. 
In~\cite{le2013mode,zamani2015symbolic}, an approach, which does not require state space discretization, has been presented for computing symbolic abstractions of incrementally stable systems. In~\cite{pola2012integrated,girard2016safety}, algorithms combining discrete controller synthesis with on-the-fly computation of symbolic abstractions have been 
developed. Compositional approaches have also been explored in several papers~\cite{tazaki2008bisimilar,reissig2010abstraction,meyer2015adhs,boskos2015decentralized,kim2015compositional,dallal2015compositional,pola2016symbolic,pola2016decentralized}.
In such approaches, a system with a control specification is decomposed into subsystems with local control specifications.
Then, for each subsystem, a symbolic abstraction  can be computed and a local controller is synthesized while assuming that the other subsystems meet their local specifications.
This approach, called \emph{assume-guarantee} reasoning~\cite{henzinger1998you}, enables the use of symbolic control techniques for higher dimensional systems.

In this paper, we develop a novel compositional approach for symbolic control synthesis for a general class of {discrete time} nonlinear systems.
Our approach clearly differs from the previously mentioned works (and particularly from our previous work~\cite{meyer2015adhs}) by the possibility for subsystems to share common state variables through the definition for each subsystem of locally modeled but uncontrolled variables, which are accessible to the local controller.
Hence, this makes it possible for local controllers to share information on some of the states of the system.
In this setting, we develop compositional approaches for computing symbolic abstractions and synthesizing controllers that maintain the state of the system in some specified safe set.

The paper is organized as follows.
Section~\ref{sec preliminaries} introduces the class of systems, safety controllers and the abstraction framework considered in the paper.
Section~\ref{sec compositional} presents a compositional approach for computing abstractions from symbolic subsystems with overlapping sets of states.
Compositional controller synthesis is addressed in Section~\ref{sec synthesis}.  
Section~\ref{sec comparison} provides results to compare abstractions and controllers obtained from different system decompositions, and a discussion on the computational complexity of the approach. Numerical experiments are then reported in Section~\ref{sec simulation}.

\section{Preliminaries}
\label{sec preliminaries}

\subsection{System description}
\label{sub preli cooperative}

We consider a class of discrete time nonlinear control systems modeled by the difference inclusion:
\begin{equation}
x(t+1)\in F(x(t),u(t)),\; t\in \N
\label{eq system}
\end{equation}
where $\N=\{0,1,2,\dots\}$, $x(t)\in\mathbb{R}^n$, $u(t)\in \mathcal U\subseteq\mathbb{R}^p$ denote the state and the control input, respectively, and the set-valued map $F:\R^n\times \mathcal U \rightarrow 2^{\R^n}$. 
System (\ref{eq system}) is discrete time; however, it encompasses sampled versions
of continuous time systems, possibly subject to disturbances  (see e.g.~\cite{meyer2015adhs,reissig2016}).

Throughout the paper, we assume, for simplicity, that for all $x\in \mathbb{R}^n$, $u\in \mathcal U$, $F(x,u)\ne \emptyset$.
For a subset of states $\mathcal X'\subseteq \R^n$  and inputs $\mathcal U' \subseteq \mathcal U$ we denote
$$
F(\mathcal X',\mathcal U') = \bigcup_{x\in\mathcal{X'},u\in\mathcal{U'}}F(x,u).
$$
Exact computation of $F(\mathcal X',\mathcal U')$ may not always be possible, especially when (\ref{eq system}) corresponds to the sampled dynamics of a continuous time system.
Therefore, we will assume throughout the paper that we are able to compute, for all sets of states $\mathcal X'\subseteq \R^n$  and of inputs $\mathcal U' \subseteq \mathcal U$,
a set $\overline{F}(\mathcal X',\mathcal U')$ verifying
\begin{equation}
F(\mathcal X',\mathcal U') \subseteq \overline{F}(\mathcal X',\mathcal U'). 
\label{eq over reachable set centralized}
\end{equation} 
Several methods exist for computing such over-approximations for linear~\cite{girard2005reachability,kurzhanskiy2007ellipsoidal,le2010reachability} and nonlinear~\cite{sassi2012reachability,althoff2014reachability,coogan2015mixed,reissig2016} systems.

\subsection{Transition systems and safety controllers}
\label{sub preli alternating}
A {\it transition system} is defined as a triple $S=(X,U,\delta)$ consisting of:
\begin{itemize}
\item a set of states $X$;
\item a set of inputs $U$;
\item a transition map $\delta : X\times U \rightarrow 2^X$.
\end{itemize}
A transition $x'\in \delta(x,u)$ means that $S$ can evolve from state $x$ to state $x'$ under input $u$.
$U(x)$ denotes the set of enabled inputs at state $x$: i.e. $u\in U(x)$ if and only if $\delta(x,u)\neq\emptyset$.
A trajectory of $S$ is a finite or infinite sequence of transitions $(x^0,u^0,x^1,u^1,\dots)$ such that $x^{t+1}\in \delta(x^t,u^t)$, for $t\in\N$

In the following, we consider a safety synthesis problem for transition system $S$: let $\mathcal X\subseteq X$ be a subset of safe states, a {\it safety controller} for system $S$ and safe set $\mathcal X$ is a map $C:X \rightarrow 2^U$
such that:
\begin{itemize}
\item for all $x\in X$, $C(x)\subseteq U(x)$;
\item its domain $dom(C)=\{x\in X |\; C(x)\ne \emptyset\}\subseteq \mathcal X$;
\item for all $x\in dom(C)$ and $u\in C(x)$, $\delta(x,u) \subseteq dom(C)$. 
\end{itemize}
Essentially, a safety controller makes it possible to generate infinite trajectories of $S$, $(x^0,u^0,x^1,u^1,\dots)$ such that $x^t \in \mathcal X$, for all $t\in \N$ as follows:
$x^0\in dom(C)$, $u^t\in C(x^t)$ and $x^{t+1}\in \delta(x^t,u^t)$, for all $t\in \N$.
It is known (see e.g.~\cite{tabuada2009symbolic}) that there exists a {\it maximal safety controller} $C^*$ for system $S$ and safe sate $\mathcal X$ such that 
for all safety controllers $C$, for all $x\in X$, it holds $C(x)\subseteq C^*(x)$.

\subsection{Feedback refinement relations}

Complex transition systems motivate the use of abstractions, since finding a control strategy for an abstraction is generally simpler than for the original system.
However, to derive a controller for the original system from that of the abstraction, the systems must satisfy a formal behavioral relationship such as alternating simulation~\cite{tabuada2009symbolic}.
In this paper, we will rely on the notion of feedback refinement relations~\cite{reissig2016}, which form a special case of alternating simulation relations:
\begin{defn}[Feedback refinement]
\label{def simulation}
Given two transition systems $S_a=(X_a,U_a,\delta_a)$ and $S_b=(X_b,U_b,\delta_b)$, with $U_b\subseteq U_a$,
a map $H:X_a\rightarrow X_b$ defines a feedback refinement relation from $S_a$ to $S_b$ if for all $(x_{a},x_{b})\in X_{a}\times X_{b}~\text{with}~x_{b}=H(x_{a})$:
\begin{itemize}
\item $U_b(x_b) \subseteq U_a(x_a)$;
\item for all $u\in U_b(x_b)$, $H(\delta_a(x_a,u))\subseteq \delta_b(x_b,u).$
\end{itemize}
We denote $S_a\preceq_{\mathcal{FR}}S_b$.
\end{defn}
In the previous definition, $S_a$ represents a complex concrete system while $S_b$ is a simpler abstraction. 
From Definition~\ref{def simulation}, it follows that all abstract inputs $u$ of $S_b$ can also be used in $S_a$ such that all concrete transitions in $S_a$ are matched by an abstract transition in $S_b$.
As a result, controllers synthesized using the abstraction $S_b$ can be interfaced with the map $H$ to obtain a controller for the concrete system $S_a$ (see~\cite{reissig2016}).
In particular, if $C_b:X_b\rightarrow 2^{U_b}$ is a safety controller for transition system $S_b$ and safe set $\mathcal X_b\subseteq X_b$, then 
$C_a:X_a\rightarrow 2^{U_a}$, given by $C_a(x_a)=C_b(H(x_a))$ for all $x_a\in X_a$, is a safety controller for transition system $S_a$ and safe set 
$\mathcal X_a=H^{-1}(\mathcal X_b) \subseteq X_a$.

\section{Compositional abstraction}
\label{sec compositional}

System (\ref{eq system}) can be described as a transition system $S=(X,U,\delta)$ where, $X=\R^n$, $U=\mathcal U$ and $\delta=F$;
let $\mathcal X \subseteq \R^n$ be a subset of states of interest. In this section, we present a compositional approach for computing symbolic abstractions of transition system $S$.

In order to allow for system decomposition, we will make the following assumption on the structure of the state and input sets $\mathcal X$ and $\mathcal U$:
\begin{assum} 
\label{assum set}
The following equalities hold:
\begin{IEEEeqnarray*}{ll}
\mathcal X= \mathcal X_1 \times \dots \times \mathcal X_{\bar n}, &  \text{ with } \mathcal X_i \subseteq \R^{n_i},\; i\in I= \{1,\dots,\bar n\};\\
\mathcal U= \mathcal U_1 \times \dots \times \mathcal U_{\bar p}, & \text{ with } \mathcal U_j \subseteq \R^{p_j},\; j\in J = \{1,\dots,\bar p\}.
\end{IEEEeqnarray*}
\end{assum}
States $x\in \R^n$ and inputs  $u\in \R^p$ can thus be seen as vectors of elementary components: $x=(x_1,\dots,x_{\bar n})$ with $x_i \in \R^{n_i}$ for $i\in I$, and
$u=(u_1,\dots,u_{\bar p})$ with $u_j \in \R^{p_j}$ for $j\in J$.

For $i\in I$, let $\mathcal P_i$ be a finite partition of the set $\mathcal X_i$, then let $\mathcal P$ be the finite partition of the safe set $\mathcal X$ obtained from the partitions $\mathcal P_i$ as follows:
\begin{equation*}
\mathcal P = \left\{ s_1 \times \dots \times s_{\bar n} |\; s_i \in \mathcal P_i, \; i\in I \right\}.
\end{equation*}
Similarly, for $j\in J$, let $\mathcal V_j$ be a finite subset of $\mathcal U_j$, then let $\mathcal V$ be the finite subset of $\mathcal U$ given by the Cartesian product of the sets $\mathcal V_j$:
\begin{equation*}
\mathcal V =   \mathcal V_1 \times \dots \times \mathcal V_{\bar p}.
\end{equation*}

\subsection{System decomposition}
\label{sub compo decomposition}

Let $m\in \N$, with $1\le m \le \min(\bar n, \bar p)$,  let $\Sigma= \{1,\dots,m\}$,
the symbolic abstraction of $S$ is obtained by composition of $m$ symbolic subsystems $S_\sigma$, $\sigma	\in \Sigma$.

In the following, we use two types of indices:
\begin{itemize}
\item Latin letters $i \in I$, $j\in J$, refer to $x_i$ and $u_j$ the components of the state and input $x$ and $u$ of system $S$. 

\item Greek letters $\sigma \in \Sigma$ refer to $S_\sigma$ the $\sigma$-th symbolic subsystem, $s_\sigma$ and $u_\sigma$ denote the state and input of system $S_\sigma$ respectively. 
\end{itemize}

We will use $\pi_i:\R^n \rightarrow \R^{n_i}$ and $\pi_j:\R^p \rightarrow \R^{p_j}$ to denote the projections over components $x_i$ and $u_j$, with $i\in I$, $j\in J$, respectively. 
For $\mathcal X' \subseteq \R^n$ and $\mathcal U' \subseteq \R^p$, we denote $\mathcal X'_i = \pi_i(\mathcal X')$ and  $\mathcal U'_j = \pi_j(\mathcal U')$. 
Similarly, for subset of indices $I'\subseteq I$, $J'\subseteq J$,
$\pi_{I'}: \R^n \rightarrow \prod_{i\in I'}  \R^{n_i}$
 and 
 $\pi_{J'}:\R^p \rightarrow \prod_{j\in J'}  \R^{p_j}$ 
 denote the projections over the set of components $\{x_i |\; i\in I'\}$ and $\{u_j |\; j\in J'\}$, respectively;
we use the notation  $x_{I'} = \pi_{I'}(x)$, $\mathcal X'_{I'} = \pi_{I'}(\mathcal X')$, $u_{J'}= \pi_{J'}(u)$ and $\mathcal U'_{J'} = \pi_{J'}(\mathcal U')$.

For $\sigma\in \Sigma$, subsystem $S_\sigma$ can be described using the following sets of indices:
\begin{itemize}
\item $I_\sigma^c \subseteq I$, with $I_\sigma^c\ne \emptyset$, denotes the state components to be controlled in $S_\sigma$, $(I_1^c,\dots,I_m^c)$ is a partition of the state indices $I$;
\item $I_\sigma \subseteq I$, with $I_\sigma^c\subseteq I_\sigma$, denotes the state components modeled in $S_\sigma$;
\item $I_\sigma^o \subseteq I$, with $I_\sigma^o=I_\sigma \backslash I_\sigma^c$, denotes the state components that are modeled but not controlled in $S_\sigma$;
\item $I_\sigma^u \subseteq I$, with $I_\sigma^u= I \backslash I_\sigma$, denotes the remaining state components that are unmodeled in $S_\sigma$;
\item $J_\sigma \subseteq J$, with $J_\sigma \ne \emptyset$, denotes the control input components modeled in $S_\sigma$, $(J_1,\dots,J_m)$ is a partition of the control input indices $J$;
\item $J_\sigma^u \subseteq J$ with $J_\sigma^u=J \backslash J_\sigma$, denotes the remaining control input components that are unmodeled in $S_\sigma$.
\end{itemize}

It is important to note that the subsystems may share common modeled state components (i.e. the sets of indices $I_\sigma$ may overlap), though the sets of controlled state components $I_\sigma^c$ and modeled control input components $J_\sigma$ are necessarily disjoints. 
Intuitively, $S_\sigma$ will be used to control state components $I_\sigma^c$ using input components $J_\sigma$; other state components $I_\sigma^o\cup I_\sigma^u$ will be controlled in other subsystems using input components $J_\sigma^u$. Though state components $I_\sigma^o$ will be controlled in other subsystems, they are modeled in $S_\sigma$ and thus information on their dynamics is available for the control of $S_\sigma$.

Let us remark that the sets of indices $I_\sigma^o$ and $I_\sigma^u$ may possibly be empty if $I_\sigma=I_\sigma^c$ and $I_\sigma=I$, respectively. 
If $m=1$, there is only one subsystem and we encompass the usual centralized abstraction approach~(see e.g.~\cite{tabuada2009symbolic,zamani2012symbolic,coogan2015mixed,reissig2016}).

\begin{remark}
In theory, the choice of the sets of indices can be made arbitrarily. However, if the considered system has some structure, i.e. if it consists of interconnected components, a natural decomposition is to associate to each component $\mathcal C$ one subsystem $S_\sigma$ where: the controlled states $I_\sigma^c$ and the modeled control input $J_\sigma$
are the states and control inputs of 
$\mathcal C$ and the modeled but uncontrolled states $I_\sigma^o$ are the states of other components that have the strongest interactions with $\mathcal C$.
\end{remark}
\subsection{Symbolic subsystems}
\label{sub compo abstraction}

Let $\sigma \in \Sigma$, the symbolic subsystem $S_\sigma$ is an abstraction of $S$, which models only state and input components $x_{I_\sigma}$ and $u_{J_\sigma}$ respectively.
Formally, subsystem $S_\sigma$ is defined as a transition system $S_\sigma=(X_\sigma,U_\sigma,\delta_{\sigma})$ where:
\begin{itemize}
\item the set of states $X_\sigma$
is a finite partition of $\pi_{I_\sigma}(\R^n)$, given by $X_\sigma=X_\sigma^0\cup\{Out_\sigma\}$ where  $Out_\sigma=\pi_{I_\sigma}(\R^n)\setminus \mathcal X_{I_\sigma}$ and
$$
X_\sigma^0= \left\{ \prod_{i\in I_\sigma} s_i \Big|\; s_i \in \mathcal P_i, \; i\in I_\sigma \right\}
$$
 is a finite partition of $\mathcal X_{I_\sigma}$;
\item the set of inputs $U_\sigma$ is a finite subset of $\mathcal U_{J_\sigma}$ given by
$$U_\sigma=\prod_{j\in J_\sigma}  \mathcal V_j.$$ 
\end{itemize}

To define the transition relation of $S_\sigma$, let us first define the following map:
given $s_\sigma \in X_\sigma^0$ and $u_\sigma\in U_\sigma$, we define the set ${\Phi}_\sigma(s_\sigma,u_\sigma)\subseteq \R^n$ as follows:
\begin{equation}
{\Phi}_\sigma(s_\sigma,u_\sigma)= \overline{F}(\mathcal X\cap \pi^{-1}_{I_\sigma}(s_\sigma),\mathcal U\cap \pi^{-1}_{J_\sigma}(\{u_\sigma\})) .
\label{eq reachable set Si}
\end{equation}
The set ${\Phi}_\sigma(s_\sigma,u_\sigma)$ is therefore an over-approximation of successors of states $x\in \mathcal X$ with $\pi_{I_\sigma}(x)\in s_\sigma$, for control inputs $u\in \mathcal U$ with $\pi_{J_\sigma}(u)=u_\sigma$.
Then, we define the transition relation of $S_\sigma$ as follows:
\begin{itemize}
\item for all $s_\sigma \in X_\sigma^0,~u_\sigma \in U_\sigma,~s_\sigma'\in X_\sigma^0$,
\begin{equation}
\label{eq trans1a}
s_{\sigma}' \in \delta_\sigma(s_{\sigma},u_{\sigma}) \iff s_\sigma'\cap\pi_{I_\sigma}({\Phi}_\sigma(s_\sigma,u_\sigma))\neq\emptyset;
\end{equation}
\item for all $s_\sigma \in X_\sigma^0,~u_\sigma\in U_\sigma$,
\begin{equation}
\label{eq trans1b}
\hspace{-0.3cm}
Out_\sigma\in \delta_\sigma(s_{\sigma},u_{\sigma}) \iff 
\left\{
\begin{array}{l}
\pi_{I_\sigma}({\Phi}_\sigma(s_\sigma,u_\sigma))\cap \mathcal X_{I_\sigma}=\emptyset\\
 \text{or } 
\pi_{I_\sigma^c}({\Phi}_\sigma(s_\sigma,u_\sigma)) \nsubseteq \mathcal X_{I_\sigma^c}.
\end{array} \right.
\end{equation}
\end{itemize}

\begin{remark}
\label{remark:subsystem1}
The first condition in (\ref{eq trans1b}) holds if and only if there does not exist any transition defined by (\ref{eq trans1a}), because $X_\sigma^0$ is a partition of $\mathcal X_{I_\sigma}$. As a consequence, it follows that for all $s_\sigma \in X_\sigma^0,~u_\sigma \in U_\sigma$, $\delta_\sigma(s_\sigma,u_\sigma)\ne \emptyset$ and thus $U_\sigma(s_\sigma)=U_\sigma$.
\end{remark}

\ifdouble
  \begin{figure}[b]
  \centering
  \includegraphics[width=0.9\columnwidth]{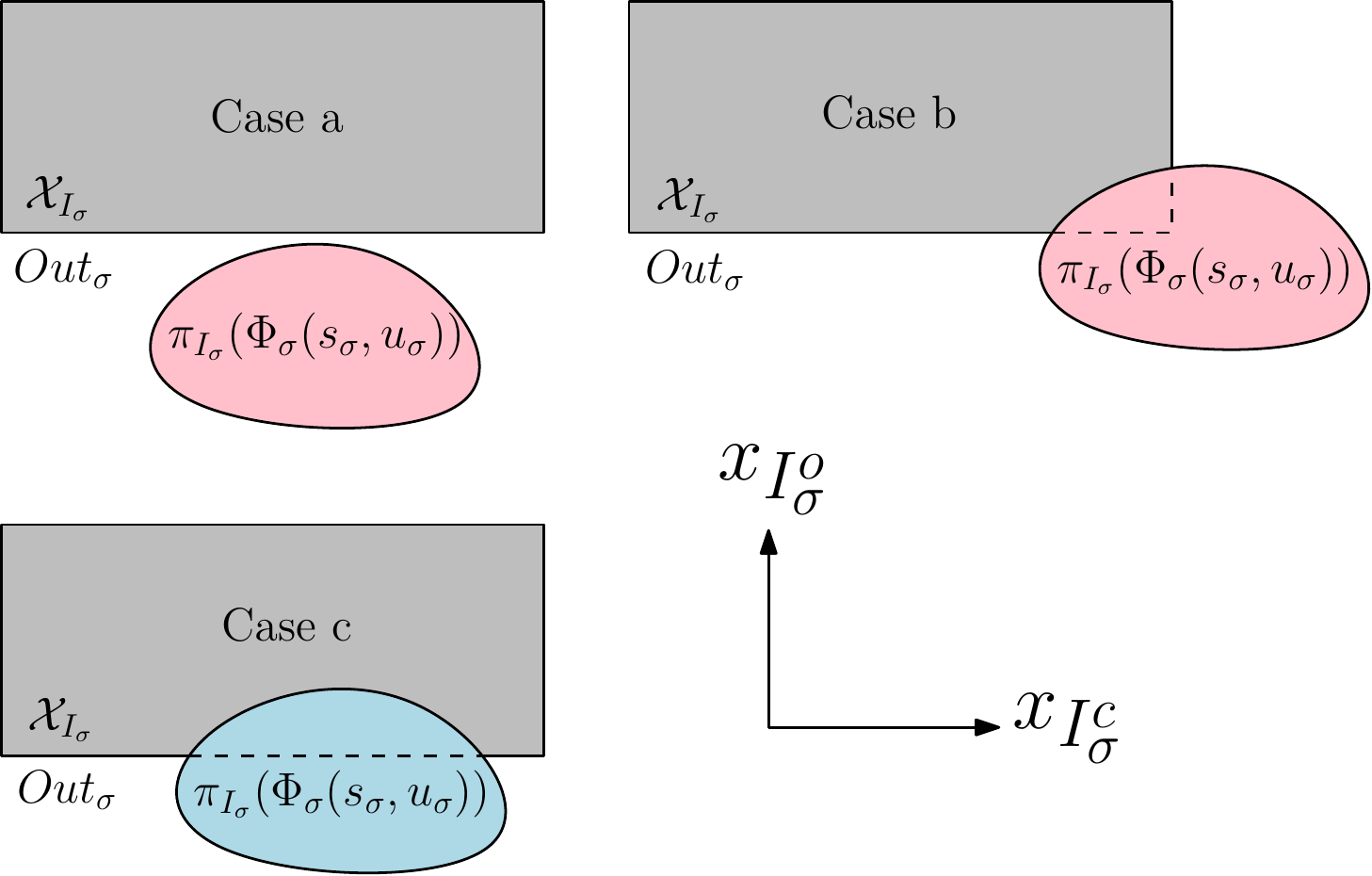}
  \caption{Illustration of (\ref{eq trans1b}): a transition towards $Out_\sigma$ is created in cases a and b, but not in case c.}
  \label{fig ag2}
  \end{figure}
\else
  \begin{figure}[h]
  \centering
  \includegraphics[width=0.6\textwidth]{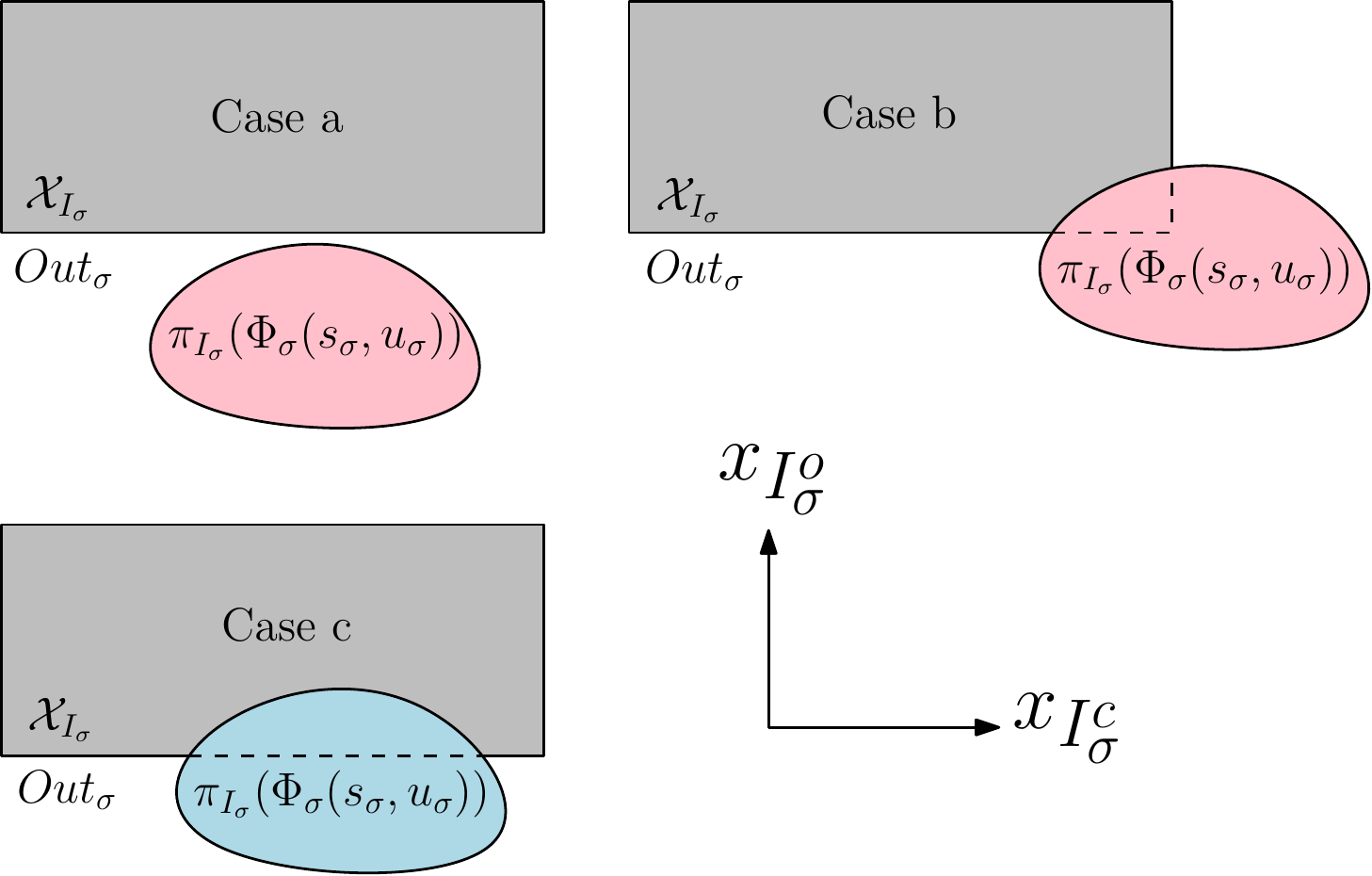}
  \caption{Illustration of (\ref{eq trans1b}): a transition towards $Out_\sigma$ is created in cases a and b, but not in case c.}
  \label{fig ag2}
  \end{figure}
\fi

\begin{remark}
\label{remark:subsystem2}
According to (\ref{eq trans1b}), a transition to $Out_\sigma$ exists if   $\pi_{I_\sigma}({\Phi}_\sigma(s_\sigma,u_\sigma))$ is entirely outside 
 $\mathcal X_{I_\sigma}$ (first condition and Figure~\ref{fig ag2}.a); or
if  $\pi_{I_\sigma^c}({\Phi}_\sigma(s_\sigma,u_\sigma))$ is not contained in $\mathcal X_{I_\sigma^c}$  (second condition and Figure~\ref{fig ag2}.b).
It should be noted that in the case where the reachable set  $\pi_{I_\sigma^c}({\Phi}_\sigma(s_\sigma,u_\sigma))$ is contained in $\mathcal X_{I_\sigma^c}$
but $\pi_{I_\sigma^o}({\Phi}_\sigma(s_\sigma,u_\sigma))$ is not contained in $\mathcal X_{I_\sigma^o}$ as in Figure~\ref{fig ag2}.c, no transition is created towards $Out_\sigma$.
Finally, if $I_\sigma=I_\sigma^c$, (\ref{eq trans1b}) becomes equivalent to
$$
Out_\sigma\in \delta_\sigma(s_{\sigma},u_{\sigma}) \iff \pi_{I_\sigma}({\Phi}_\sigma(s_\sigma,u_\sigma)) \nsubseteq \mathcal X_{I_\sigma},
$$
which is the condition used in~\cite{meyer2015adhs}, for compositional abstractions where the set of modeled state components $I_\sigma$ do not overlap
(i.e. $I_\sigma=I^c_\sigma$, for all $\sigma\in \Sigma$).
\end{remark}

\subsection{Composition}
\label{sub compo alternating}

In this section, we show how the previous subsystems $S_\sigma$, with $\sigma \in \Sigma$, can be composed in order to define a symbolic abstraction $S_c$ of the original system $S$.
The main result of the section is Theorem~\ref{th simulation Sc}, which shows that there exists  a feedback refinement relation from $S$ to $S_c$.

The composition of the subsystems $S_\sigma$, $\sigma \in \Sigma$, is given by the transition system $S_c=(X_c,U_c,\delta_c)$ where:
\begin{itemize}
\item the set of states $X_c$
is a finite partition of $\R^n$, given by $X_c=X_c^0\cup\{Out\}$ where  $Out=\R^n\setminus \mathcal X$ and
$X_c^0=\mathcal{P}$ is a finite partition of $\mathcal X$;
\item the set of inputs $U_c=\mathcal V$ is a finite subset of $\mathcal U$.
\end{itemize}

Let us remark that by definition of $X_c^0$ and $X_\sigma^0$, we have that for all $s\in X_c^0$, its projection $s_{I_\sigma} \in  X_\sigma^0$.
Similarly, for all $u \in U_c$, its projection $u_{J_\sigma} \in  U_\sigma$.
The transition relation of $S_c$ can therefore be defined as follows:
\begin{itemize}
\item for all $s\in X_c^0,~u\in U_c,~s'\in X_c^0$,
\begin{equation}
\label{eq trans2a}
s'\in \delta_c(s,u) \Longleftrightarrow\forall \sigma \in \Sigma,~s'_{I_\sigma} \in\delta_\sigma(s_{I_\sigma},u_{J_\sigma});
\end{equation}

\item for all $s\in X_c^0,~u\in U_c$,
\begin{equation}
\label{eq trans2b}
Out \in \delta_c(s,u) \Longleftrightarrow \exists \sigma \in \Sigma,\; Out_\sigma \in\delta_\sigma(s_{I_\sigma},u_{J_\sigma}).
\end{equation}
\end{itemize}

\begin{remark}
Because the sets of modeled state components $I_\sigma$ are allowed to overlap, the transition relation of $S_c$ cannot simply be obtained as the Cartesian product of the transition relations of the subsystems $S_\sigma$, as in~\cite{meyer2015adhs}.
Indeed, for $s\in X_c^0,~u\in U_c$, it is possible that for all $\sigma \in \Sigma$, there exists $s'_\sigma \in X_\sigma^0$, such that $s_{\sigma}' \in \delta_\sigma(s_{I_\sigma},u_{J_\sigma})$.
However, a transition to $X_c^0$ will exist in $S_c$ if and only if there exists $s'\in  X_c^0$ such that $s_{I_\sigma}'=s'_\sigma$, for all $\sigma \in \Sigma$. 
\end{remark}

In view of the previous remark, it is legitimate to ask if the composition of the subsystems can lead to couples of states and inputs $(s,u)\in X_c^0 \times U_c$ without a successor. The following proposition shows that this is not the case:
\begin{prop}
\label{prop input composed}
Under Assumption~\ref{assum set}, for all $s\in X_c^0$ we have $U_c(s)=U_c$, i.e.\ $\delta_c(s,u)\neq\emptyset$, for all $u\in U_c$.
\end{prop}

\begin{proof}
Let $s\in X_c^0$ and $u\in U_c$. Then for all $\sigma \in \Sigma$, $s_{I_\sigma}\in X_\sigma^0$, $u_{J_\sigma}\in U_\sigma$ and by construction, $\delta_\sigma(s_{I_\sigma},u_{J_\sigma})\neq\emptyset$ (see Remark~\ref{remark:subsystem1}).
If there exists a subsystem $S_\sigma$ such that $Out_\sigma \in \delta_\sigma(s_{I_\sigma},u_{J_\sigma})$, then by definition of $S_c$ we have $Out\in \delta_c(s,u)$.
Otherwise, we have that $Out_\sigma \notin \delta_\sigma(s_{I_\sigma},u_{J_\sigma})$ for all $\sigma \in \Sigma$,
which from the second condition of (\ref{eq trans1b}) implies that 
\begin{equation}
\label{eq p1a}
\forall \sigma\in \Sigma,\; \pi_{I^c_\sigma}({\Phi}_\sigma(s_{I_\sigma},u_{J_\sigma}))\subseteq \mathcal X_{I_\sigma^c}. 
\end{equation}
Remarking that $s\subseteq \mathcal X\cap \pi^{-1}_{I_\sigma}(s_{I_\sigma})$ and $\{u\}\subseteq \mathcal U\cap \pi^{-1}_{J_\sigma}(\{u_{J_\sigma}\})$,
the following inclusion follows from (\ref{eq over reachable set centralized}) and (\ref{eq reachable set Si}):
\ifdouble
  \begin{eqnarray}
  \nonumber
  F(s,\{u\}) &\subseteq & F(\mathcal X\cap \pi^{-1}_{I_\sigma}(s_{I_\sigma}),\mathcal U\cap \pi^{-1}_{J_\sigma}(\{u_{J_\sigma}\}))\\
  \label{eq p1b}
  &\subseteq & {\Phi}_\sigma(s_{I_\sigma},u_{J_\sigma}).
  \end{eqnarray}
\else
  \begin{equation}
  \label{eq p1b}
  F(s,\{u\}) \subseteq  F(\mathcal X\cap \pi^{-1}_{I_\sigma}(s_{I_\sigma}),\mathcal U\cap \pi^{-1}_{J_\sigma}(\{u_{J_\sigma}\}))  \subseteq  {\Phi}_\sigma(s_{I_\sigma},u_{J_\sigma}).
  \end{equation}
\fi
Therefore, from (\ref{eq p1a}) and (\ref{eq p1b}) it follows
$$
\forall \sigma\in \Sigma,\; \pi_{I^c_\sigma}(F(s,\{u\}) )\subseteq \mathcal X_{I_\sigma^c}. 
$$
This, together with Assumption~\ref{assum set} and the fact that $(I_1^c,\dots,I_m^c)$ is a partition of $I$, implies that $F(s,\{u\}) \subseteq \mathcal X$. 
Since $X_c^0$ is a partition of $\mathcal X$, there exists $s'\in X_c^0$ such that $s'\cap F(s,\{u\})\neq\emptyset$. 
 Then, (\ref{eq p1b}) gives $s'\cap {\Phi}_\sigma(s_{I_\sigma},u_{J_\sigma}) \neq\emptyset$, for all $\sigma\in \Sigma$.
 Thus, for all $\sigma\in \Sigma$, $s'_{I_\sigma} \cap \pi_{I_\sigma}({\Phi}_\sigma(s_{I_\sigma},u_{J_\sigma})) \neq \emptyset$. 
It follows from (\ref{eq trans1a}) that $s_{I_\sigma}'\in \delta_\sigma(s_{I_\sigma},u_{J_\sigma})$ for all $\sigma\in \Sigma$, which gives, by (\ref{eq trans2a}), $s'\in \delta_c(s,u)$.
\end{proof}

We can now state the main result of the section:
\begin{thm}
\label{th simulation Sc}
Let  the map $H:X \rightarrow X_c$ be given by
$H(x)=s$ if and only if $x\in s$. Then,
under Assumption~\ref{assum set}, $H$
defines a feedback refinement relation from $S$ to $S_c$: $S\preceq_{\mathcal{FR}}S_c$.
\end{thm}
\begin{proof}
Let $s\in X_c^0$, $x\in s$, $u\in U_c(s)=U_c\subseteq U=U(x)$, $x'\in \delta(x,u)=F(x,u)$ and $s'=H(x')$.
Since $x\in s$, we have $x'\in F(s,\{u\})$. Then, let us consider the two possible cases:
\begin{itemize}
\item $x'\in \mathcal X$ -- We have by (\ref{eq p1b}), $x'\in {\Phi}_\sigma(s_{I_\sigma},u_{J_\sigma})$, for all $\sigma\in \Sigma$.
Since $x'\in \mathcal X$, then $s'\in X_c^0$, it follows from $x'\in s'$ that $s'\cap  {\Phi}_\sigma(s_{I_\sigma},u_{J_\sigma})\neq \emptyset$, for all $\sigma\in \Sigma$.
 Then, for all $\sigma \in \Sigma$, $s'_{I_\sigma}\in X_\sigma^0$ and 
$s_{I_\sigma}'\cap  \pi_{I_\sigma} ({\Phi}_\sigma(s_{I_\sigma},u_{J_\sigma}))\neq \emptyset$.
From (\ref{eq trans1a}),
$s_{I_\sigma}'\in \delta_\sigma(s_{I_\sigma},u_{J_\sigma})$, for all $\sigma\in \Sigma$ and by (\ref{eq trans2a}) we have $s'\in \delta_c(s,u)$.

\item $x'\notin \mathcal X$ -- Then, $F(s,\{u\}) \not\subseteq \mathcal X$. Then, from Assumption~\ref{assum set} and the fact that $(I_1^c,\dots,I_m^c)$ is a partition of $I$, it follows that there exists $\sigma \in \Sigma$ such that 
$\pi_{I_\sigma^c}(F(s,\{u\})) \not\subseteq \mathcal X_{I_\sigma^c}$. From (\ref{eq p1b}), we have $\pi_{I_\sigma^c}(
{\Phi}_\sigma(s_{I_\sigma},u_{J_\sigma})) \not\subseteq \mathcal X_{I_\sigma^c}$. Then, from  (\ref{eq trans1b}), $Out_\sigma \in \delta_\sigma(s_{I_\sigma},u_{J_\sigma})$, and from
(\ref{eq trans2b}), $Out \in \delta_c(s,u)$. Since $x'\notin \mathcal X$, $s'=Out$.
\end{itemize}
The case $s=Out$ trivially satisfies Definition~\ref{def simulation} since $U_c(Out)=\emptyset$ by definition of $S_c$.
\end{proof}

Note that the composed abstraction $S_c$ is only created in this section to prove the feedback refinement relationship but one should avoid computing it in practice since it would defeat the purpose of the compositional approach. 
We end the section by stating an instrumental result, which will be used in Section~\ref{sec comparison} when comparing abstractions obtained from different system decompositions.

\begin{lemma}
\label{prop ag2}
Under Assumption~\ref{assum set}, for all $s\in X_c^0$ and $u\in U_c$, $Out \in \delta_c(s,u)$ if and only if there exists $\sigma \in \Sigma$ such that $\pi_{I_\sigma^c}(\Phi_\sigma(s_{I_\sigma},u_{J_\sigma}))\not\subseteq \mathcal X_{I_\sigma^c}$.
\end{lemma}
\begin{proof}
Sufficiency is straightforward from (\ref{eq trans1b}) and (\ref{eq trans2b}). As for necessity,
if $Out\in \delta_c(s,u)$, then there exists a subsystem $\sigma$ such that $Out_\sigma\in \delta_\sigma(s_{I_\sigma},u_{J_\sigma})$.
From (\ref{eq trans1b}), either  $\pi_{I_\sigma^c}(\Phi_\sigma(s_{I_\sigma},u_{J_\sigma}))\not\subseteq \mathcal X_{I_\sigma^c}$ (in which case the property holds), or
$\pi_{I_\sigma^c}(\Phi_\sigma(s_{I_\sigma},u_{J_\sigma})) \subseteq \mathcal X_{I_\sigma^c}$ and 
$\pi_{I_\sigma}(\Phi_\sigma(s_{I_\sigma},u_{J_\sigma})) \cap \mathcal X_{I_\sigma} =\emptyset$.
Then, by  Assumption~\ref{assum set} and since $I^o_\sigma=I_\sigma \setminus I_\sigma^c$, it follows that $\pi_{I^o_\sigma}(\Phi_\sigma(s_{I_\sigma},u_{J_\sigma})) \cap \mathcal X_{I^o_\sigma} =\emptyset$. Then, 
by (\ref{eq p1b}), $\pi_{I^o_\sigma}(F(s,\{u\})) \cap \mathcal X_{I^o_\sigma} =\emptyset$.
Thus, it follows that $\pi_{I^o_\sigma}(F(s,\{u\})) \nsubseteq \mathcal X_{I^o_\sigma}$.
From Assumption~\ref{assum set}, there exists $i\in I_\sigma^o$, such that  $\pi_{i}(F(s,\{u\})) \nsubseteq \mathcal X_{i}$.
Then, let $\sigma' \in \Sigma$ such that $i\in I_{\sigma'}^c$, then $\pi_{I_{\sigma'}^c}(F(s,\{u\})) \nsubseteq \mathcal X_{I_{\sigma'}^c}$.
By  (\ref{eq p1b}), it follows that $\pi_{I_{\sigma'}^c}(\Phi_\sigma(s_{I_{\sigma'}},u_{J_{\sigma'}})) \nsubseteq \mathcal X_{I_{\sigma'}^c}$ and the property holds.
\end{proof}

\section{Compositional safety synthesis}
\label{sec synthesis}

In this section, we consider the problem of synthesizing a safety controller for transition system $S$ and safe set $\mathcal X$.
Because of the feedback refinement relation from $S$ to $S_c$, this can be done by solving the safety synthesis problem for transition system $S_c$ and safe set $X_c^0$.
We propose a compositional approach, which works on the symbolic subsystems $S_\sigma$ and does not require computing the composed abstraction $S_c$.

For $\sigma \in \Sigma$, let $C_\sigma^*:X_\sigma \rightarrow 2^{U_\sigma}$ be the maximal safety controller for transition system $S_\sigma$ and safe set $X_\sigma^0$. Since $S_\sigma$ has only finitely many states and inputs, $C_\sigma^*$ can be computed in finite time using a fixed point algorithm~\cite{tabuada2009symbolic}.
Now, let the controller $C_c:X_c \rightarrow 2^{U_c}$ be defined by $C_c(Out)=\emptyset$ and
\begin{equation}
\label{eq Cc}
\forall s\in X_c^0,\; C_c(s)=\{u \in U_c |\; u_{J_\sigma} \in C^*_\sigma(s_{I_\sigma}), \forall \sigma \in \Sigma \}.
\end{equation}

\begin{thm}
\label{th safety composition} Under Assumption~\ref{assum set},
$C_c$ is a safety controller for transition system $S_c$ and safe set $X_c^0$.
\end{thm}

\begin{proof} 
From Proposition~\ref{prop input composed} and since $C_c(Out)=\emptyset$, it is clear that for all $s\in X_c$, we have $C_c(s) \subseteq U_c(s)$.
$C_c(Out)=\emptyset$ also gives $dom(C_c)\subseteq X_c^0$. Then, let $s\in dom(C_c)\subseteq X_c^0$, $u\in C_c(s)$ and $s'\in \delta_c(s,u)$.
If $s'\notin X_c^0$, then $s'=Out$ and from (\ref{eq trans2b}), there exists $\sigma \in \Sigma$, such that $Out_\sigma \in \delta_\sigma(s_{I_\sigma},u_{J_\sigma})$, which contradicts the fact that $u_{J_\sigma} \in C^*_\sigma(s_{I_\sigma})$ with $C^*_\sigma$ safety controller for transition system $S_\sigma$ and safe set $X_\sigma^0$. 
Hence, we necessarily have $s'\in X_c^0$, and from (\ref{eq trans2a}), it follows that $s'_{I_\sigma} \in \delta_\sigma(s_{I_\sigma},u_{J_\sigma})$, for all $\sigma\in \Sigma$. 
Moreover, $u_{J_\sigma} \in C^*_\sigma(s_{I_\sigma})$ gives that $s'_{I_\sigma} \in dom(C^*_\sigma)$. 
Then for all $\sigma \in \Sigma$, let $u'_{\sigma} \in C^*_\sigma(s'_{I_\sigma})$. Since $(J_1,\dots,J_m)$ is a partition of $J$, there exists $u' \in U_c$ such that $u'_{J_\sigma}=u'_{\sigma}$
for all $\sigma \in \Sigma$. Then, by (\ref{eq Cc}), $u'\in C_c(s')$ and thus $s'\in dom(C_c)$. It follows that $C_c$ is a safety controller for transition system $S_c$ and safe set $X_c^0$.
\end{proof}

\begin{remark}
Since the sets of modeled state components $I_\sigma$ may overlap, it is in principle possible that $dom(C_c)=\emptyset$ while $dom(C_\sigma^*)\ne \emptyset$, for all $\sigma \in \Sigma$. The reason is that an element of $dom(C_c)$ is obtained from states in $dom(C_\sigma^*)$, which coincide on their common modeled states, as shown in (\ref{eq Cc}).

\end{remark}

\subsection{Particular case: non-overlapping state sets}
\label{sec particular cases}

Though $C_\sigma^*$ is a maximal safety controller for all $\sigma\in \Sigma$, the safety controller $C_c$ is generally not maximal. 
Maximality can be obtained when the set of modeled states $I_\sigma$, $\sigma\in \Sigma$ do not overlap (or equivalently when for all $\sigma \in \Sigma$, $I_\sigma^c=I_\sigma$).
In that case,  the following result holds:
\begin{prop}
\label{pro safety max}
Under Assumption~\ref{assum set}, let $I_\sigma^c=I_\sigma$, for all $\sigma\in \Sigma$.
Then, $C_c$ is the maximal safety controller for transition system $S_c$ and safe set $X_c^0$.
\end{prop}

\begin{proof} Let $C_c' : X_c \rightarrow 2^{U_c}$ be a safety controller for transition system $S_c$ and safe set $X_c^0$.
For $\sigma\in \Sigma$, let the controllers $C_\sigma': X_\sigma \rightarrow 2^{U_\sigma}$ be defined by $C_\sigma'(Out_\sigma)=\emptyset$ and
for all $s_\sigma \in X_\sigma^0$, 
\begin{equation}
\label{eq proj safe}
C_\sigma'(s_\sigma)=\{u_\sigma \in \pi_{J_\sigma}(C_c'(s)) |\; s\in X_c^0, \; s_{I_\sigma}=s_\sigma\}.
\end{equation}
Let us show that $C_\sigma'$ is a safety controller for system $S_\sigma$ and safe set $X_\sigma^0$.
Following Remark~\ref{remark:subsystem1}, and since $C_\sigma'(Out_\sigma)=\emptyset$, it is clear that for all $s_\sigma \in X_\sigma$, we have $C_\sigma'(s_\sigma) \subseteq U_\sigma(s_\sigma)$.
$C_\sigma'(Out_\sigma)=\emptyset$ also gives $dom(C_\sigma')\subseteq X_\sigma^0$. 
Then, let $s_\sigma\in dom(C_\sigma')$, $u_\sigma \in C_\sigma'(s_\sigma)$ and $s_\sigma' \in \delta_\sigma(s_\sigma,u_\sigma)$, let us prove that $s'_\sigma \in dom (C_\sigma')$.
By (\ref{eq proj safe}), there exists $s\in dom(C_c')$ and $u\in C_c'(s)$ such that $s_{I_\sigma}=s_\sigma$ and $u_{J_\sigma}=u_\sigma$. 
Since $C_c'$ is a safety controller, $\delta_c(s,u)\subseteq dom(C_c') \subseteq X_c^0$.
Moreover, since  the sets $I_\sigma$ are not overlapping, it follows from (\ref{eq trans2a}) that 
there exists $s'\in \delta_c(s,u)$ such that $s'_{I_\sigma}=s'_\sigma$.
Then, $s'\in dom(C_c')$ and (\ref{eq proj safe}) give that $s'_\sigma \in dom (C_\sigma')$.
Hence  $C_\sigma'$ is a safety controller for system $S_\sigma$ and safe set $X_\sigma^0$.
Then, by maximality of $C_\sigma^*$, it follows that for all $s_\sigma\in X_\sigma^0$, $C_\sigma'(s_\sigma) \subseteq C_\sigma^*(s_\sigma)$.
Finally,
let $s\in dom(C_c')$ and $u\in C_c'(s)$, then by (\ref{eq proj safe}), $u_{J_\sigma} \in C_\sigma'(s_{I_\sigma})\subseteq C_\sigma^*(s_{I_\sigma})$ for all $\sigma \in \Sigma$.
By (\ref{eq Cc}), $u\in C_c(s)$, which shows the maximality of $C_c$.
\end{proof}

\section{Comparisons}
\label{sec comparison}

In this section, we provide theoretical comparisons between abstractions and controllers given by the previous approach using two different system decompositions.

In addition to the set of state and input indices defined in Section~\ref{sub compo decomposition},  let us consider partitions $(\hat I_1^c,\dots,\hat  I_{\hat m}^c)$ and $(\hat J_1,\dots,\hat J_{\hat m})$ of the state and input indices and subsets of state indices $(\hat I_1,\dots, \hat I_{\hat m})$ with $\hat I_{\hat \sigma}^c \subseteq \hat I_{\hat \sigma}$, for all ${\hat \sigma} \in \hat \Sigma=\{1,\dots,{\hat m}\}$. We define the same objects as before (i.e. subsystems, abstraction, controllers, etc.) for this system decomposition and denote them with hatted notations. We make the following assumption on the two system decompositions under consideration.

\begin{assum}
\label{assum set inclusion}
There exists a surjective map $\gamma: \hat \Sigma \rightarrow \Sigma$ such that, 
for all $\hat \sigma \in \hat \Sigma$ and $\sigma=\gamma(\hat\sigma) \in \Sigma$,
$$
\hat I^c_{\hat \sigma} \subseteq I^c_\sigma, \;  \hat I_{\hat \sigma} \subseteq I_\sigma, \; \hat J_{\hat \sigma} \subseteq J_\sigma.
$$
\end{assum}
From the previous assumption, and since $(\hat I_1^c,\dots,\hat  I_{\hat m}^c)$ and $(\hat J_1,\dots,\hat J_{\hat m})$ are partitions of the state and input indices, we have that
\begin{equation}
\label{eq set inclusion}
\forall \sigma \in \Sigma, \bigcup_{\hat \sigma \in \gamma^{-1} (\sigma)} \hat I^c_{\hat \sigma} = I^c_\sigma\; \text{ and }  \bigcup_{\hat \sigma \in \gamma^{-1} (\sigma)} \hat J_{\hat \sigma} = J_\sigma.
\end{equation} 

In addition, we will make the following mild assumption on the over-approximations of the reachable sets:
\begin{assum}
\label{assum reachable set comparison}
For all $\mathcal X'' \subseteq \mathcal X' \subseteq \R^n$, $\mathcal U'' \subseteq \mathcal U' \subseteq \mathcal U$, the following inclusion holds
$$
\overline{F}(\mathcal X'',\mathcal U'') \subseteq \overline{F}(\mathcal X',\mathcal U').
$$
\end{assum}
This assumption can be shown to be satisfied by most existing techniques for over-approximating the reachable set, and in particular by those mentioned in Section~\ref{sub preli cooperative}.
In addition, under Assumptions~\ref{assum set inclusion} and~\ref{assum reachable set comparison}, it follows from (\ref{eq reachable set Si}), that for all $\hat \sigma \in \hat \Sigma$ and $\sigma=\gamma(\hat\sigma) \in \Sigma$,
\begin{equation}
\label{eq reachable set comparison}
\forall s\in \mathcal P,\; u\in \mathcal V,\; \Phi_\sigma(s_{I_\sigma},u_{J_\sigma}) \subseteq \hat \Phi_{\hat \sigma}(s_{\hat I_{\hat \sigma}},u_{\hat J_{\hat \sigma}}) .
\end{equation}

\subsection{Abstractions}
\label{sub compare alternating}

We start by comparing the compositional abstractions $S_c$ and $\hat S_c$ resulting from the two different decompositions:
\begin{thm}
\label{prop alternating sufficient}
Under Assumptions~\ref{assum set}, \ref{assum set inclusion} and~\ref{assum reachable set comparison}, the identity map is a feedback refinement relation from $S_c$ to $\hat S_c$:
$ S_c \preceq_{\mathcal{FR}} \hat S_c$.
\end{thm}
\begin{proof}
Let us first remark that  $X_c^0= \hat X_c^0$, $X_c = \hat X_c$ and $U_c=\hat U_c$.
Then, from Proposition~\ref{prop input composed}, for all $s\in X_c^0 = \hat X_c^0$, $U_c(s)=U_c=\hat U_c=\hat U_c(s)$. 
Since $\hat U_c(Out)=\emptyset$, Definition \ref{def simulation} holds if $\delta_c(s,u)\subseteq \hat{\delta}_c(s,u)$, for all $s \in X_c^0$, $u\in U_c$.
 
Hence, let  $s \in X_c^0$, $u\in U_c$ and $s'\in \delta_c(s,u)$, then let us consider the two possible cases:
\begin{itemize}
\item $s'\in X_c^0$ -- We have by (\ref{eq trans2a}) and (\ref{eq trans1a}) that 
$$
\forall \sigma \in \Sigma,\; s_{I_\sigma}'\cap\pi_{I_\sigma}(\Phi_\sigma(s_{I_\sigma},u_{J_\sigma}))\neq\emptyset.
$$
Then, from (\ref{eq reachable set comparison}), follows that
$$
\forall \hat \sigma \in \hat \Sigma \text{ and } \sigma=\gamma(\hat \sigma),\;  s_{I_\sigma}'\cap\pi_{I_\sigma}(\hat \Phi_{\hat \sigma}(s_{\hat I_{\hat \sigma}},u_{\hat J_{\hat \sigma}}))\neq\emptyset.
$$
By Assumption~\ref{assum set inclusion}, $\hat I_{\hat \sigma} \subseteq I_\sigma$, for all $\hat \sigma \in \hat \Sigma$ and $\sigma=\gamma(\hat \sigma)$.
Thus it follows that
$$
\forall \hat \sigma \in \hat \Sigma,\;  s_{\hat I_{\hat \sigma}}'\cap\pi_{\hat I_{\hat \sigma}}(\hat \Phi_{\hat \sigma}(s_{\hat I_{\hat \sigma}},u_{\hat J_{\hat \sigma}}))\neq\emptyset.
$$
Then, from  (\ref{eq trans1a}) and (\ref{eq trans2a}), we have $s'\in  \hat{\delta}_c(s,u)$.
\item $s'=Out$ -- From Lemma~\ref{prop ag2}, we know that there exists $\sigma \in \Sigma$, such that $\pi_{I_\sigma^c}(\Phi_\sigma(s_{I_\sigma},u_{J_\sigma}))\nsubseteq\mathcal X_{I_\sigma^c}$. Then, from Assumption~\ref{assum set}, there exists $i\in I_\sigma^c$ such that $\pi_{i}(\Phi_\sigma(s_{I_\sigma},u_{J_\sigma}))\nsubseteq\mathcal X_{i}$. From (\ref{eq set inclusion}), there exists $\hat \sigma \in \hat \Sigma$, such that $\sigma=\gamma(\hat \sigma)$ and $i\in \hat I_{\hat \sigma}^c$.
From (\ref{eq reachable set comparison}), it follows that $\pi_{i}(\hat \Phi_{\hat \sigma}(s_{\hat I_{\hat \sigma}},u_{\hat J_{\hat \sigma}}))\nsubseteq\mathcal X_{i}$
and  $\pi_{\hat I_{\hat \sigma}^c}(\hat \Phi_{\hat \sigma}(s_{\hat I_{\hat \sigma}},u_{\hat J_{\hat \sigma}}))\nsubseteq\mathcal X_{\hat I_{\hat \sigma}^c}$.
Then, from  (\ref{eq trans1b}) and (\ref{eq trans2b}), we have $Out \in  \hat{\delta}_c(s,u)$.
\end{itemize}
\end{proof}

Note that the conditions in the previous Theorem are only sufficient conditions, since depending on the dynamics of the system, a feedback refinement relation could also exist between two unrelated decompositions (in terms of index set inclusion).  

\begin{remark} Theorem~\ref{prop alternating sufficient} gives an indication on how one should modify the sets of indices to reduce the conservatism of the compositional symbolic abstraction. Firstly, one can keep the same number of subsystems and the same controlled states $I_\sigma^c$ and modeled control input $J_\sigma$, while 
considering additional modeled but uncontrolled states in $I_\sigma^o$. Secondly, one can merge two or more subsystems by merging their controlled states, modeled control inputs and modeled but uncontrolled states.
\end{remark}

\subsection{Controllers}
\label{sub compare safety}

We now compare the controllers  obtained by the approach described in Section~\ref{sec synthesis}. 
The comparison of controllers is more delicate than the comparison of abstractions and we shall need the additional assumption that the sets of indices $I_\sigma$ do not overlap (note that the sets $\hat I_{\hat \sigma}$ may still overlap).

\begin{corollary}
\label{pro safety comparison}
Under Assumptions~\ref{assum set}, \ref{assum set inclusion} and~\ref{assum reachable set comparison}, let $I_\sigma^c=I_\sigma$, for all $\sigma\in \Sigma$.
Then, for all $s\in X_c$, $\hat C_c(s)\subseteq C_c(s)$.
\end{corollary}

\begin{proof} From Theorem~\ref{th safety composition}, $\hat C_c$ is a safety controller for system $\hat S_c$ and safe set $\hat X_c$.
From Theorem~\ref{prop alternating sufficient}, it follows that $\hat C_c$ is also a safety controller for system $S_c$ and safe set $X_c$.
Then, by Proposition~\ref{pro safety max}, the maximality of $C_c$ gives us for all $s\in X_c$, $\hat C_c(s)\subseteq C_c(s)$.
\end{proof}

Let us remark that the assumption that the sets of indices $I_\sigma$ do not overlap is instrumental in the proof since it uses Proposition~\ref{pro safety max}.
The question whether similar results hold in the absence of such assumption is an open question, which is left as future research.

\subsection{Complexity}
\label{sub compare complexity}

In this section, we discuss the computational complexity of the approach and show the advantage of using a compositional approach rather than a centralized one.
Let $|.|$ denote the cardinality of a set.

The computation of symbolic subsystem $S_\sigma$ requires a number of reachable set approximations equal to $\prod_{i\in I_\sigma} |\mathcal P_i| \times \prod_{j\in J_\sigma} |\mathcal V_j|$, each creating up to $(1+\prod_{i\in I_\sigma} |\mathcal P_i|)$ successors.
This results in an overall time and space complexity $\mathcal C_1$ of computing all symbolic subsystems $S_\sigma$, $\sigma \in \Sigma$:
$$
\mathcal C_1 =\mathcal O \Big(\sum_{\sigma\in \Sigma} \big(\prod_{i\in I_\sigma} |\mathcal P_i|^2 \times 
 \prod_{j\in J_\sigma} |\mathcal V_j |\big)\Big).
$$

The computation of the safety controller $C_\sigma$ by a fixed point algorithm requires a number of iteration which is bounded by the number of states in the safe set $X_\sigma^0$: $\prod_{i\in I_\sigma} |\mathcal P_i|$.
The complexity order of computing an iteration can be bounded by the number of transitions in $S_\sigma$. This results in an overall time and space complexity $\mathcal C_2$ of computing all safety controllers $C_\sigma$, $\sigma \in \Sigma$:
$$
\mathcal C_2 =\mathcal O \Big(\sum_{\sigma\in \Sigma} \big(\prod_{i\in I_\sigma} |\mathcal P_i|^3 \times 
 \prod_{j\in J_\sigma} |\mathcal V_j |\big)\Big).
$$

To illustrate the advantage of using a compositional approach, let us consider two extremal cases in the particular case where the number of state and input component are equal
$I=J$. The centralized case corresponds to $\Sigma=\{1\}$, with $I_1=J_1=I$. In that case the complexity of the overall approach is of order
$\mathcal O \Big( \prod_{i\in I}  |\mathcal P_i|^3 \times| \mathcal V_i | \Big).$  
The fully decentralized case corresponds to $\Sigma=I=J$, with $I_\sigma=J_\sigma=\{\sigma\}$ for all $\sigma \in \Sigma$. In that case the complexity of the overall approach is of order
$\mathcal O \Big( \sum_{i\in I}  |\mathcal P_i|^3 \times| \mathcal V_i |  \Big).$
Hence, one can see that while the complexity of the centralized approach is exponential in the number of state and input components $|I|$, it becomes linear with the fully decentralized approach. Intermediate decompositions enable to balance the computational complexity and the conservativeness of the approach, in view of the discussions in Sections~\ref{sub compare alternating} and~\ref{sub compare safety}.

\section{Numerical illustration}
\label{sec simulation}
In this section, we illustrate the results of this paper on the temperature regulation in a circular building of $n\geq 3$ rooms, each equipped with a heater.
For each room $i\in\{1,\dots,n\}$, the variations of the temperature $T_i$ are described by the discrete-time model adapted from~\cite{pola2016decentralized}:
\begin{equation*}
\label{eq simulation}
T_i^+=T_i+\alpha(T_{i+1}+T_{i-1}-2T_i)+\beta(T_e-T_i)+\gamma(T_h-T_i)u_i,
\end{equation*}
where $T_{i+1}$ and $T_{i-1}$ are the temperature of the neighbor rooms (with $T_0=T_n$ and $T_{n+1}=T_1$), $T_e=-1\,^\circ C$ is the outside temperature, $T_h=50\,^\circ C$ is the heater temperature, $u_i\in[0,0.6]$ is the control input for room $i$ and the conduction factors are given by $\alpha=0.45$, $\beta=0.045$ and $\gamma=0.09$.
This model can be proved to be monotone as defined in~\cite{angeli_monotone}, which allows us to use efficient algorithms for over-approximating the reachable sets~\cite{meyer2015adhs,coogan2015mixed}.
Moreover, the over-approximation scheme satisfies Assumption~\ref{assum reachable set comparison}.

The safe set $\mathcal X$ is given by a $n$-dimensional interval (specified later) which is uniformly partitioned into $\lambda_T$ intervals per component (for a total of $\lambda_T^n$ symbols in $\mathcal P$) and the control set $\mathcal U=[0,0.6]^n$ is uniformly discretized into $\lambda_u$ values per component (for a total of $\lambda_u^n$ values in $\mathcal V$).
We consider $3$ possible system decompositions, which provides us with $3$ different abstractions:
\begin{itemize}
\item $S_c^1$, the centralized case (i.e. $m=1$), with a single subsystem containing all states and controls, with $I_1^1=I_1^{c1}=J_1^1=\{1,\dots,n\}$;
\item $S_c^2$, a general case from Section~\ref{sec compositional} with $m=n$ subsystems, $I_\sigma^{c2}= J_\sigma^2=\{\sigma\}$ and $I_\sigma^2=\{\sigma-1,\sigma,\sigma+1\}$ for all $\sigma\in\Sigma=\{1,\dots,m\}$;
\item $S_c^3$, a case with $m=n$ subsystems and non-overlapping state sets as in Section~\ref{sec particular cases}, with $I_\sigma^3=I_\sigma^{c3}=J_\sigma^3=\{\sigma\}$ for all $\sigma\in\Sigma$.
\end{itemize}
Both $S_c^2$ and $S_c^3$ have one subsystem per room, but  subsystems of $S_c^3$ only focus on the state and control of the considered room, while subsystems of $S_c^2$ also model (but do not control) the temperatures of both neighbor rooms.

Since Assumption~\ref{assum set inclusion} holds for both pairs $(S_c^1,S_c^2)$ and $(S_c^2,S_c^3)$, Theorem~\ref{prop alternating sufficient} immediately gives the feedback refinements $S_c^1 \preceq_{\mathcal{FR}}S_c^2 \preceq_{\mathcal{FR}}S_c^3$.
Corollary~\ref{pro safety comparison} also holds for the pair $(S_c^1,S_c^2)$ since $I_1^{c1}=I_1^1$,
but it is not guaranteed to hold for the pair $(S_c^2,S_c^3)$ since $I_\sigma^{c2}\neq I_\sigma^2$.
In the following, we report numerical results in two different conditions.
The numerical implementation has been done using Matlab on a laptop with a $2.6$ GHz CPU and $8$ GB of RAM.
\medskip

{\bf Case 1:} {$n=4$, $\mathcal X=[17,22]\times[19,22]\times[20,23]\times[20,22]$.}

The abstractions and syntheses are generated in the $6$ cases corresponding to state partitions with $\lambda_T\in\{5,10,20\}$ and input discretizations with $\lambda_u\in\{3,4\}$.
Table~\ref{tab 1} reports the cardinalities $|\mathcal{P}|=\lambda_T^4$ of the state partition $\mathcal{P}$, and $|dom(C_c)|$ of the domain of the safety controllers for each abstraction $S_c^1$, $S_c^2$ and $S_c^3$.
Table~\ref{tab 3} reports the computation times (in seconds) required to create the abstractions and synthesize safety controllers on all subsystems of $S_c^1$, $S_c^2$ and $S_c^3$.

\begin{table}[!t]
\begin{center}
\begin{tabular}{|c|c||c||c|c|c|}
\hline
$\lambda_T$ & $\lambda_u$ & $|\mathcal{P}|=\lambda_T^4$ & $|dom(C_c^1)|$ & $|dom(C_c^2)|$ & $|dom(C_c^3)|$\\
\hline
$5$ & $3$ & $625$ & $525$ & $0$ & $0$\\
$5$ & $4$ & $625$ & $525$ & $0$ & $0$\\
$10$ & $3$ & $10000$ & $8900$ & $8710$ & $0$\\
$10$ & $4$ & $10000$ & $8900$ & $8710$ & $0$\\
$20$ & $3$ & $160000$ & $145180$ & $143480$ & $0$\\
$20$ & $4$ & $160000$ & $145180$ & $143480$ & $0$\\
\hline
\end{tabular}
\medskip
\caption{\label{tab 1} Number of elements in the domains of the safety controllers for the safe set $\mathcal X=[17,22]\times[19,22]\times[20,23]\times[20,22]$.}
\end{center}
\end{table}

\begin{table}[!t]
\begin{center}
\begin{tabular}{|c|c||c|c|c|}
\hline
$\lambda_T$ & $\lambda_u$ & $S_c^1$ & $S_c^2$ & $S_c^3$\\
\hline
$5$ & $3$ & $1.80$ & $0.17$ & $0.07$\\
$5$ & $4$ & $5.49$ & $0.20$ & $0.07$\\
$10$ & $3$ & $64$ & $0.46$ & $0.06$\\
$10$ & $4$ & $210$ & $0.56$ & $0.06$\\
$20$ & $3$ & $6044$ & $2.87$ & $0.09$\\
$20$ & $4$ & $18339$ & $3.84$ & $0.44$\\
\hline
\end{tabular}
\medskip
\caption{\label{tab 3} Computation times (in seconds) for the safe set $\mathcal X=[17,22]\times[19,22]\times[20,23]\times[20,22]$.}
\end{center}
\end{table}

We check numerically that Theorem~\ref{prop alternating sufficient} and Corollary~\ref{pro safety comparison} hold.
In particular, in these conditions, the safety inclusion $C_c^3(s)\subseteq C_c^2(s)$ for all $s\in\mathcal{P}$ trivially holds due to $dom(C_c^3)=\emptyset$, although Corollary~\ref{pro safety comparison} could not provide theoretical guarantees in this case.

Two main conclusions on the proposed compositional approach can be obtained from Tables~\ref{tab 1} and~\ref{tab 3}.
Firstly, while the compositional case without state overlap (as in $S_c^3$, Section~\ref{sec particular cases} and~\cite{meyer2015adhs}) fails to synthesize safety controllers, the general case allowing state overlaps (as in $S_c^2$ and Section~\ref{sec compositional}) provides significantly better safety results for a relatively small addition to the computation time.
Secondly, the compositional approach with state overlaps $S_c^2$ requires a negligible computation time compared to the large computational cost of the centralized approach $S_c^1$ (e.g.\ in the last row of Table~\ref{tab 3}, we need less than $4$ \emph{seconds} for $S_c^2$ and more than $5$ \emph{hours} for $S_c^1$), while still obtaining similar safety results as long as the state partition $\mathcal{P}$ is not too coarse.

In addition to having more information in each subsystem of $S_c^2$ compared to those in the non-overlapping case $S_c^3$, the better safety results in $S_c^2$ can also be explained by the shapes that can be taken by the domain of the safety controllers with each approach.
On the one hand, the safety domain $dom(C_c^3)$ in the non-overlapping case $S_c^3$ can only take the form of a hyper-rectangle in $\R^4$ since it is obtained by the Cartesian product of the one-dimensional safety domains $dom(C_\sigma^3)$ of its subsystems.
On the other hand, the general case with state overlaps $S_c^2$ is more permissive since its subsystems $S_\sigma^2$ have a three-dimensional state space, thus allowing more complicated shapes of their safety domains $dom(C_\sigma^2)$ as displayed in Figure~\ref{fig simu} for subsystem $\sigma=4$.
$S_c^2$ thus has more chances finding a safety domain compatible with the considered system dynamics and control objective.
\medskip

\ifdouble
  \begin{figure}[tbh]
  \centering
  \includegraphics[width=1\columnwidth]{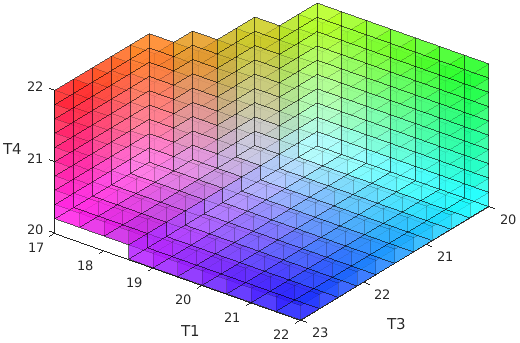}
  \caption{Visualization of the domain $dom(C_\sigma^2)$ of the safety controller for subsystem $\sigma=4$ of $S_c^2$. Each axis is associated with one component of the RGB color model to facilitate the visualization of depth.}
  \label{fig simu}
  \end{figure}
\else
  \begin{figure}[tbh]
  \centering
  \includegraphics[width=0.6\textwidth]{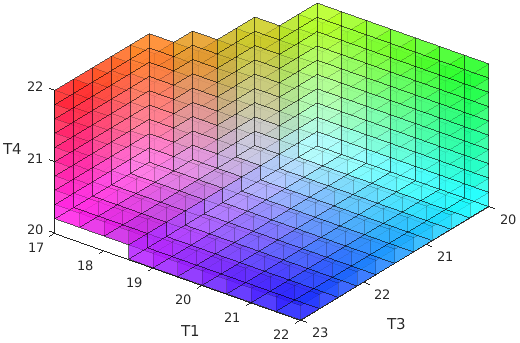}
  \caption{Visualization of the domain $dom(C_\sigma^2)$ of the safety controller for subsystem $\sigma=4$ of $S_c^2$. Each axis is associated with one component of the RGB color model to facilitate the visualization of depth.}
  \label{fig simu}
  \end{figure}
\fi

{\bf Case 2:} {$n=20$, $\mathcal X=[19,21]^{20}$, $\lambda_T=10$, $\lambda_u=5$.}

A second example is proposed to demonstrate the scalability of the compositional approach in a $20$-room building.
Note that the safe set $\mathcal X=[19,21]^{20}$ is only chosen homogeneous in all rooms for convenience of notation, and the proposed approach is still applicable for other safe sets.
Since this case is clearly out of reach from the centralized approach of $S_c^1$, we focus on the compositional abstractions $S_c^2$ and $S_c^3$ with and without state overlaps, respectively.

For the non-overlapping case of $S_c^3$, the total computation time is $0.12$ second and the resulting safety controller is empty ($dom(C_c^3)=\emptyset$).
For the case with state overlaps of $S_c^2$, the total computation time is $3.04$ seconds and the resulting safety controller covers the whole safe set $\mathcal{X}$ ($dom(C_c^2)=\mathcal{P}$).
Therefore, in addition to the scalability of both these compositional approaches, this simulation also confirms the conclusions of the previous example that the method with state overlaps provides significantly better safety results at a reduced computational cost.
We also obtain similarly low computation times while not having to rely on the homogeneity of the specifications as it is the case in~\cite{pola2016decentralized}.

\section{Conclusion}
\label{sec conclusion}

In this paper, we presented a new compositional approach for symbolic controller synthesis.
The dynamics are decomposed into subsystems that give a partial description of the global model.
It is remarkable that the sets of states of subsystems can overlap.
Symbolic abstractions can be computed for each subsystem, and a local safety controller can be synthesized such that the composition of the obtained controllers is proved to realize the global safety specification.
Numerical experiments demonstrate the significant complexity reduction compared to centralized approaches and the advantages obtained from the introduction of state overlaps in the subsystems.

Future work will focus on extending the approach to other types of specifications such as reachability or more general properties specified by automata or temporal logic formula.

\bibliographystyle{abbrv}
\bibliography{Meyer_compositional}

\end{document}